\documentclass{amsart}
\usepackage{amsmath,amssymb,pxfonts}
\usepackage{todonotes}
\newtheorem{thm}{Theorem}[section]
\newtheorem{lem}{Lemma}[section]
\newtheorem{cor}{Corollary}[section]
\title[Weak Convergence Rate under Rough Volatility]{{\Large O}n the {\Large W}eak {\Large C}onvergence {\Large R}ate
in the {\Large D}iscretization of {\Large R}ough {\Large V}olatility
{\Large M}odels}
\author[C. Bayer, M. Fukasawa \& S. Nakahara]{{\large C}hristian {\large
B}ayer, {\large M}asaaki {\large F}ukasawa and {\large S}honosuke
{\large N}akahara}
\address{M. Fukasawa, Graduate School of Engineering Science, Osaka University,
560-8531, Japan}
   
\begin{document}
\maketitle
   \begin{abstract} 
    We study the weak convergence rate in the discretization of rough
    volatility models.
    After showing a lower bound $2H$ under a general model,
      where
    $H$ is the Hurst index of the volatility process,
    we give a sharper bound $H + 1/2$ under a linear model.
   \end{abstract}

   \section{Introduction and the main results}
The financial derivative pricing  practice amounts to computing an expectation of 
a functional of an asset price process. When the asset price process is
modelled by a diffusion, 
 the computation is either by solving an
associated PDE or by the Monte-Carlo method with a discretized process.
The error in the computation of the expectation due to the
discretization is called the weak error,
in contrast to the strong error quantifying pathwise deviation.
It is now well-known that for regular diffusion processes the weak error
in the
Euler-Maruyama discretization is  $O(n^{-1})$, where $n$ is the number of
the discretization steps. See \cite{TT,BT,K} among others.

Recently a class of stochastic volatility models called the rough
volatility models has attracted attention.
The volatility process under those models is not a diffusion but has a rougher path.
This is the only class of continuous price models that are consistent to
a power-law type term structure of implied volatility skew typically
observed in equity option markets. The volatility processes of the major
stock indices are also statistically estimated to be rough. See~\cite{GJR,FTW,F21} for the details.

For example, under the rough Bergomi model introduced by \cite{BFG}, the
asset price $S$ follows
\begin{equation*}
 \begin{split}
& S_T = S_0 \exp\left\{\int_0^T \sqrt{V_t}\left\{\rho \mathrm{d}W_t + \sqrt{1-\rho^2}\mathrm{d}W^\perp_t\right\}-
  \frac{1}{2}\int_0^T V_t \mathrm{d}t\right\},\\
  & V_t = V_0(t)\exp\left\{\int_0^t k(t,s)\mathrm{d}W_s - \int_0^t
  k(t,s)^2 \mathrm{d}s\right\},
 \end{split}
\end{equation*}
where $(W,W^\perp)$ is a $2$ dimensional standard Brownian motion,
 $V_0(t) = \mathsf{E}[V_t]$ is a deterministic function called the forward variance curve,
$k(t,s) = \eta\sqrt{2H}(t-s)^{H-1/2}$, and
$H \in (0,1/2)$, $\rho \in (-1,1)$ and  $\eta >0 $ are parameters.
The European option price with payoff $f$ is expressed as
\begin{equation*}
 \mathsf{E}[f(S_T)] = \mathsf{E}\left[
				 F\left(\int_0^T\sqrt{V_t}\mathrm{d}W_t,
				   \int_0^T V_t \mathrm{d}t
				  \right)
				\right],
\end{equation*}
where
\begin{equation*}
 F(x,y) = \int_\mathbb{R} f\left(S_0\exp\left\{\rho x +
					 \sqrt{y(1-\rho^2)}z -
					 \frac{1}{2}y \right\}\right)\phi(z)\mathrm{d}z,
\end{equation*}
where $\phi$ is the standard normal density.
For these non-diffusion models the finite dimensional PDE pricing is not possible
anymore.
The Monte-Carlo pricing with a discretized stochastic integral
\begin{equation*}
 \int_0^1 \sqrt{V_{\frac{[nt]}{n}}} \mathrm{d}W_t =
  \sum_{i=0}^{n-1}\sqrt{V_{\frac{i}{n}}} (W_{\frac{i+1}{n}}-W_{\frac{i}{n}})
\end{equation*}
is proposed in \cite{BFG}, while the convergence rate of the the
   associated weak error is still an open problem.

Motivated by the rough volatility modelling, we are interested in the
weak error rate in the discretization of a stochastic integral
\begin{equation*}
 X_1 := \int_0^1 \sigma(Y_t,t)\mathrm{d}W_t,  \ \
  Y_t := \int_0^t (t-s)^{H-1/2}\mathrm{d}W_s
\end{equation*}
for a sufficiently regular function $\sigma$.
Under the rough Bergomi model, $\sqrt{V_t} = \sigma(Y_t,t)$ with
\begin{equation*}
 \sigma(y,t) = \sqrt{V_0(t)}\exp\left\{
			  \frac{\eta \sqrt{2H}}{2} y - \frac{\eta^2}{4}t^H
			 \right\}.
\end{equation*}
Noting that $Y$ is Gaussian and so the random sequence
$(W_{\frac{i}{n}},Y_{\frac{i}{n}})$, $i=0,1,\dots$ can be efficiently sampled using the Cholesky decomposition of the covariance matrix, 
a numerically feasible discretized integral is given by
\begin{equation*}
X^n_1 := \int_0^1 \sigma\left(Y_{\frac{[nt]}{n}},\frac{[nt]}{n}\right) \mathrm{d}W_t.
\end{equation*}
The associated weak error with respect to a test function $f$ is given
by $\varepsilon_n(f,1)$, where
\begin{equation*}
\varepsilon_n(f,a) =  \mathsf{E}[f(X_1)] - \mathsf{E}[f((1-a)X_1 + aX^n_1)]
\end{equation*}
for $a \in [0,1]$.

  Recently Bayer et al.~\cite{BHT} has considered the case $\sigma(y,t) = y$ and showed that the weak error rate is at least $H+1/2$, that is, $\varepsilon_n(f,1) = O(n^{-H-1/2})$ for $f \in C^k_b$ with $k \geq 1/H$.
The approach taken by \cite{BHT} is based on an infinite dimensional Markov representation of $Y$.
With a finite dimensional Markov approximation to $Y$, the existing techniques for diffusions are then applicable.
In this article we take a different approach for the same problem. We apply the duality approach introduced by \cite{K} that is based on the Malliavin calculus and known to be effective to deal with non-Markov models.

Our first result is for a general time homogeneous $\sigma(y,t)=\sigma(y)$ but not sharp.

\begin{thm}\label{th1}
   If $\sigma \in C^2_p$ and $f \in C^{3}_p$, then
 $\varepsilon_n(f,a) = O(n^{-2H})$ uniformly in $a \in [0,1]$.
\end{thm}

The second result is sharper but only for $\sigma$ being linear as in \cite{BHT}.
 \begin{thm}\label{th2}
  If $\sigma(y) = y$ and $f \in C^{2m+1}_p$, then
  $\varepsilon_n(f,1) = O(n^{- ((2mH) \wedge (H+1/2))})$.
 \end{thm}
  \begin{cor}
     If $\sigma(y) = y$ and
 $f \in C^k_p$ with $k \geq 2+1/(2H)$, then
    $\varepsilon_n(f,1) = O(n^{-(H+1/2)})$.  
  \end{cor}
 In comparison with \cite{BHT}, we allow test functions of polynomial growth, and for $H< 1/4$ we require less regularity on test functions to ensure the same rate. The proof is shorter as we see below.

  It is still an open problem whether the rate $H+1/2$ is sharp even for the linear model. 
  See \cite{BHT} for some numerical experiments that suggest it is the case.
  It would be straightforward to extend Theorem~\ref{th1} to include 
  the case that the volatility function $\sigma$ is time-dependent.  It seems however difficult to extend
   our proof of Theorem~\ref{th2} to include
  the case that $\sigma$ is nonlinear;
  the linearity results in a recursive structure of the weak error which we exploit as a key.

  \section{The duality approach}
We give the proofs of Theorem~\ref{th1} and \ref{th2} in this section. Let
  \begin{equation*}
 k_n(s,t) = 
(t-s)^{H-1/2} - \left(\frac{[nt]}{n}-s\right)_+^{H-1/2}.
\end{equation*}
Here and in the sequel, we mean by $(x)_+^\alpha$
\begin{equation*}
    (x)_+^\alpha = x^\alpha 1_{(0,\infty)}(x)
\end{equation*}
for $\alpha \in \mathbb{R}$ and $x \in \mathbb{R}$.
 We start with a lemma. 
  \begin{lem}\label{lem1}
For any $\alpha \in (-1/2-H,1/2-H)$,
\begin{equation*}
  \int_0^1 \int_0^t (t-s)^\alpha |k_n(s,t)|
\mathrm{d}s\mathrm{d}t  = O(n^{-\alpha-H-1/2})
\end{equation*}
   and
   \begin{equation*}
  \int_0^1\int_0^t \left(\frac{[nt]}{n}-s\right)_+^\alpha |k_n(s,t)|
\mathrm{d}s\mathrm{d}t  = O(n^{-\alpha-H-1/2}).
\end{equation*}
\end{lem}
 \begin{proof}
Changing variable as $t-s = (t-[nt]/n)u$, 
\begin{equation*}
\begin{split}
& \int_0^1\int_0^t (t-s)^\alpha|k_n(s,t)|\mathrm{d}s\mathrm{d}t \\
&= \int_0^1\left(t - \frac{[nt]}{n}\right)^{\alpha + H +1/2}
\int_0^{t/(t-[nt]/n)}
u^{\alpha}\left|u^{H-1/2} - (u-1)_+^{H-1/2}\right|\mathrm{d}u
\mathrm{d}t.
\end{split}
\end{equation*}
This is $O(n^{-\alpha-H-1/2})$ as claimed since
\begin{equation*}
    \int_0^2
    u^\alpha \left| u^{H-1/2} - (u-1)_+^{H-1/2}\right| <\mathrm{d} u < \infty
    \end{equation*}
for $\alpha > -1/2 -H$    and
\begin{align*}
    \int_2^\infty u^\alpha \left| u^{H-1/2} - (u-1)^{H-1/2}\right| \mathrm{d}u &= \frac{1}{|H-1/2|} \int_2^\infty u^\alpha \int_{u-1}^u v^{H-3/2} \mathrm{d}v \mathrm{d}u\\
    &\le \frac{2^{3/2-H}}{|H-1/2|} \int_2^\infty u^{\alpha+H-3/2} \mathrm{d}u \\ &= \frac{2^{1+\alpha}}{(1/2-H)(1/2-\alpha-H)}
\end{align*}
for $\alpha < 1/2-H$,
using that $v^{H-3/2}\le (u-1)^{H-3/2}$ for $u-1 \le v \le u$ and $$\sup_{u\ge 2} \frac{(u-1)^{H-3/2}}{u^{H-3/2}} = 2^{3/2-H}.$$

  The latter is similarly proved by a variable-change
  $[nt]/n-s = (t-[nt]/n)u$.
 \end{proof}
 \subsection{Proof of Theorem~\ref{th1}}
 We take the duality approach introduced by \cite{K}.
We have
\begin{equation*}
\begin{split}
 X_1 - X^n_1  = &  \int_0^1 \sigma_n(t) (Y_t -
 Y_{[nt]/n})\mathrm{d}W_t 
=  \int_0^1 \sigma_n(t)\int_0^tk_n(s,t)\mathrm{d}W_s
 \mathrm{d}W_t,
\end{split}
\end{equation*}
where
\begin{equation*}
\sigma_n(t) = \int_0^1 \sigma^\prime((1-\theta)Y_{[nt]/n} + \theta Y_t
  )\mathrm{d}\theta. 
\end{equation*}
Let
\begin{equation*}
 F_n(a) = \int_0^a f^\prime((1-b)X_1 + bX^n_1 ) \mathrm{d}b.
\end{equation*}
Then,
using the duality formula (the Malliavin integration-by-parts formula; see~\cite{Nualart} for the details), we have
\begin{equation*}
\begin{split}
 \varepsilon_n(f,a)  = &  \mathsf{E}[F_n(a)(X_1-X^n_1)] \\
= & \mathsf{E}[\int_0^1 D_tF_n(a)
 \sigma_n(t)\int_0^tk_n(s,t)\mathrm{d}W_s \mathrm{d}t]
\\
=  & \int_0^1\int_0^t \mathsf{E}[D_s(D_tF_n(a)
\sigma_n(t))] k_n(s,t)\mathrm{d}s\mathrm{d}t,
\end{split}
\end{equation*}
where $D$ denotes the Malliavin-Shigekawa derivative.
Note that
\begin{equation*}
 \begin{split}
& D_tF_n(a) = 
    \int_0^a f^{(2)}((1-b)X_1 + bX^n_1)((1-b)D_tX_1 + b D_tX^n_1)
  \mathrm{d}b,\\
 &D_tX_1 = \sigma(Y_t) + \int_t^1
  \sigma^\prime(Y_u)(u-t)^{H-1/2}\mathrm{d}W_u,\\
&  D_tX^n_1 = \sigma(Y_{[nt]/n}) + \int_t^1
  \sigma^\prime(Y_{[nu]/n})\left(\frac{[nu]}{n}-t\right)_+^{H-1/2}\mathrm{d}W_u,\\
  &
  D_s\sigma_n(t) = \int_0^1 \sigma^{(2)}((1-\theta)Y_{[nt]/n} + \theta
 Y_t)
 \left((1-\theta)\left(\frac{[nt]}{n}-s\right)_+^{H-1/2} +
 \theta(t-s)_+^{H-1/2}\right) \mathrm{d}\theta,
 \end{split}
\end{equation*}
\begin{equation}\label{DDF}
 \begin{split}
  D_s&D_t F_n(a) \\
= &  
\int_0^a f^{(2)}((1-b)X_1 + bX^n_1)((1-b)D_sD_tX_1 + b D_sD_tX^n_1)
  \mathrm{d}b\\
 & + \int_0^a f^{(3)}((1-b)X_1 + bX^n_1)
((1-b)D_sX_1 + bD_sX^n_1) 
((1-b)D_tX_1 + bD_tX^n_1)\mathrm{d}b,
 \end{split}
\end{equation}
and for $s < t$,
\begin{equation*}
 \begin{split}
  D_sD_t X_1 = &\sigma^\prime(Y_t)(t-s)^{H-1/2} + \int_t^1
  \sigma^{(2)}(Y_u)(u-s)^{H-1/2}(u-t)^{H-1/2}\mathrm{d}W_u,\\ 
  D_sD_t X^n_1 =
  &\sigma^\prime(Y_{[nt]/n})\left(\frac{[nt]}{n}-s\right)_+^{H-1/2} \\ &
  + \int_t^1
  \sigma^{(2)}(Y_{[nu]/n})\left(\frac{[nu]}{n}-s\right)_+^{H-1/2}\left(\frac{[nu]}{n}-t\right)_+^{H-1/2}\mathrm{d}W_u.
 \end{split}
\end{equation*}
Therefore, denoting by $C$ a generic constant which is independent of $a$,
we have by Minkowski's integral inequality that  
for any $n$ and $0 < s < t < 1$,
\begin{equation*}
\begin{split}
 |\mathsf{E}[D_s(D_tF_n(a)\sigma_n(t))]|& \leq 
 \|D_sD_tF_n(a)\|_2\|\sigma_n(t)\|_2
 + \|D_tF_n(a)\|_2 \|D_s\sigma_n(t)\|_2 \\
 & \leq 
C\left( \|D_sD_tF_n(a)\|_2 + \|D_tF_n(a)\|_2\left(
|t-s|^{H-1/2} + \left(\frac{[nt]}{n}-s\right)_+^{H-1/2}\right)\right)\\
&\leq 
C\Biggl(\|D_sD_tX_1\|_4 + \|D_sD_tX^n_1\|_4 
\\ & \hspace*{1cm} +
(\|D_sX_1\|_6 + \|D_sX^n_1\|_6)
(\|D_tX_1\|_6 + \|D_tX^n_1\|_6)
\\ & \hspace*{1cm} +
(\|D_tX_1\|_4 + \|D_tX^n_1\|_4)
\left(
|t-s|^{H-1/2} + \left(\frac{[nt]}{n}-s\right)_+^{H-1/2}
\right)
\Biggl) \\
& \leq 
C\left(
|t-s|^{H-1/2} + \left(\frac{[nt]}{n}-s\right)_+^{H-1/2}
\right).    
\end{split}
\end{equation*}
The result then
follows from Lemma~\ref{lem1}  with  $\alpha = H-1/2$.\qed

\subsection{Proof of Theorem~\ref{th2}}
When  $\sigma(y) = y$, we have $\sigma_n(t) = 1$, 
\begin{equation*}
 D_tX_1 = \int_0^1 |t-s|^{H-1/2}\mathrm{d}W_s, \ \ 
D_sD_t X_1 = |t-s|^{H-1/2}
\end{equation*}
and
\begin{equation*}
 D_tX^n_1 = \int_0^1\left(\frac{[nt]}{n}-s\right)_+^{H-1/2}
+  \left( \frac{[ns]}{n} - t\right)_+^{H-1/2}\mathrm{d}W_s
\end{equation*}
with $D_sD_t X^n_1 = ([nt]/n-s)_+^{H-1/2} + ([ns]/n-t)_+^{H-1/2}$.
By \eqref{DDF},
\begin{equation*}
\begin{split}
 & D_sD_t  F_n(a)    \\ &=   f^{(2)}(X_1)\int_0^a((1-b)D_sD_tX_1 +
  bD_sD_tX^n_1)\mathrm{d}b  \\
& + \int_0^a (f^{(2)}((1-b)X_1 + bX^n_1)-f^{(2)}(X_1))
 ((1-b)D_sD_tX_1 + bD_sD_tX^n_1)\mathrm{d}b \\
 & + \int_0^a f^{(3)}((1-b)X_1 + bX^n_1)
((1-b)D_sX_1 + bD_sX^n_1) 
((1-b)D_tX_1 + bD_tX^n_1)\mathrm{d}b.
\end{split}
\end{equation*}
Therefore, we have
\begin{equation*}
 \varepsilon_n(f,a) = r_{n,1}(f,a) + r_{n,2}(f,a) + r_{n,3}(f,a),
\end{equation*}
where
\begin{equation*}
 \begin{split}
&  r_{n,1}(f,a) = \mathsf{E}[f^{(2)}(X_1)] 
\int_0^a g_n(b)\mathrm{d}b,\\
&r_{n,2}(f,a) = \int_0^a \mathsf{E}[
(f^{(2)}((1-b)X_1 + bX^n_1)-f^{(2)}(X_1))] 
 g_n(b) \mathrm{d}b, \\
&g_n(b) = \int_0^1\int_0^t \left((1-b)(t-s)^{H-1/2} + b
  \left(\frac{[nt]}{n}-s\right)_+^{H-1/2}\right)
  k_n(s,t)\mathrm{d}s\mathrm{d}t,\\
&r_{n,3}(f,a) =  \int_0^a\int_0^1\int_0^t
h_n(s,t,b)k_n(s,t)
\mathrm{d}s\mathrm{d}t\mathrm{d}b,\\
&h_n(s,t,b) = \mathsf{E}[f^{(3)}((1-b)X_1 + bX^n_1)
((1-b)D_sX_1 + bD_sX^n_1) 
((1-b)D_tX_1 + bD_tX^n_1)].
 \end{split}
\end{equation*}
Notice that
\begin{equation*}
    r_{n,2}(f,a) = - \int_0^a \varepsilon_n(f^{(2)},b)
 g_n(b) \mathrm{d}b
\end{equation*}
and that $r_{n,3}(f,a) = O(n^{-H-1/2})$ uniformly in
$a \in [0,1]$ for $f \in C^3_p$
since
\begin{equation*}
\int_0^1\int_0^t
|k_n(s,t)|\mathrm{d}s\mathrm{d}t
= O(n^{-H-1/2})
\end{equation*}
by Lemma~\ref{lem1}.
Therefore,
\begin{equation*}
 \varepsilon_n(f,a) = E[f^{(2)}(X_1)]\int_0^a g_n(b)\mathrm{d}b
 -\int_0^a \varepsilon_n(f^{(2)},b)g_n(b)\mathrm{d}b + O(n^{-H-1/2}).
\end{equation*}
By induction, we obtain
 \begin{equation}\label{ind}
 \begin{split}
      \varepsilon_n(f,1) = & \sum_{i=1}^{m-1}(-1)^{i-1}
      \mathsf{E}[f^{(2i)}(X_1)]\int_0^1 \dots \int_0^{b_{i-1}} \prod_{j=1}^{i}
      g_n(b_j)\mathrm{d}b_j 
    \\ &+  (-1)^{m-1}\int_0^1\int_0^{b_1}\dots\int_0^{b_{m-2}}
   \varepsilon_n(f^{(2(m-1))},b_{m-1}) \prod_{i=1}^{m-1}g_n(b_i)\mathrm{d}b_i + O(n^{-H-1/2})
 \end{split}
 \end{equation}
 for $m \geq 2$ and $f \in C_p^{2m-1}$.
 We are going to show that the first sum is $O(n^{-1})$.
 Note that
 \begin{equation*}
     \int_0^1 \dots \int_0^{b_{i-1}} \prod_{j=1}^{i}
      g_n(b_j)\mathrm{d}b_j 
      = \frac{1}{i!} \left(\int_0^1 g_n(b)\mathrm{d}b\right)^i
 \end{equation*}
and that
\begin{equation*}
\begin{split}
\int_0^1 g_n(b)\mathrm{d}b &=
\frac{1}{2}\int_0^1\int_0^t
((t-s)^{H-1/2}  + ([nt]/n-s)_+^{H-1/2})
k_n(s,t)\mathrm{d}s\mathrm{d}t 
\\
&=\frac{1}{2}\int_0^1\int_0^t
((t-s)^{2H-1}  -([nt]/n-s)_+^{2H-1})
\mathrm{d}s\mathrm{d}t\\
&= \frac{1}{4H}\int_0^1 \left(t^{2H} -
 \left(\frac{[nt]}{n}\right)^{2H}\right)\mathrm{d}t \\
 &=\frac{1}{2} \int_0^1\int_{[nt]/n}^ts^{2H-1}\mathrm{d}s\mathrm{d}t= O(n^{-1}).
\end{split}
\end{equation*}
Therefore, from \eqref{ind}, we have
\begin{equation*}
      \varepsilon_n(f,1) =  (-1)^{m-1}\int_0^1\int_0^{b_1}\dots\int_0^{b_{m-2}}
   \varepsilon_n(f^{(2(m-1))},b_{m-1}) \prod_{i=1}^{m-1}g_n(b_i)\mathrm{d}b_i + O(n^{-H-1/2}).
 \end{equation*}
By Lemma~\ref{lem1}, we have
$g_n(b) = O(n^{-2H})$ uniformly in $b \in [0,1]$.
Since $f \in C_p^{2m+1}$ by the assumption, we have 
$\varepsilon_n(f^{(2(m-1))},b_{m-1}) = O(n^{-2H})$ uniformly in $b_{m-1}$ by
Theorem~\ref{th1}, and thus
 we conclude. \qed

\end{document}